\documentclass{ifacconf}
\usepackage{graphicx}      
\usepackage{natbib}       
\usepackage{amsmath,amssymb,amsfonts}
\usepackage{algorithm}
\usepackage{algpseudocode}
\usepackage{textcomp}
\usepackage[T1]{fontenc}
\usepackage{bm}
\usepackage{color}
\usepackage{subfigure}
\usepackage{ulem}
\usepackage{url}
\usepackage{array,makecell,booktabs}
\usepackage{threeparttable}
\usepackage{multirow}
\usepackage{listings}
\newcommand{\oomit}[1]{}
\newcommand{\argmax}{\mathop{\mathrm{argmax}}}

\newtheorem{assumption}{Assumption}
\newtheorem{definition}{Definition}
\newtheorem{proposition}{Proposition}

\newtheorem{remark}{Remark}
\newtheorem{theorem}{Theorem}
\newtheorem{example}{Example}
\newtheorem{problem}{Problem}
\begin{document}
\begin{frontmatter}

\title{Convex Computations for Controlled Safety Invariant Sets of Black-box Discrete-time Dynamical Systems} 

\author{Taoran Wu$^{1}$, Yiling Xue$^{1,2}$, Jingduo Pan$^{1}$, Dejin Ren$^{1}$} 
\author{Arvind Easwaran$^{3}$, and Bai Xue$^{1,2}$}

\address{1. KLSS and SKLCS, ISCAS, Beijing, China; University of Chinese Academy of Sciences, Beijing, China}
\address{2. School of Advanced Interdisciplinary Sciences, University of Chinese Academy of Sciences, Beijing, China}
\address{3. College of Computing and Data Science, Nanyang Technological University, Singapore}
\address{Email: \{wutr,xuebai\}@ios.ac.cn; arvinde@ios.ac.cn.}

\begin{abstract}
Identifying controlled safety invariant sets (CSISs) is essential for safety-critical systems. This paper addresses the problem of computing CSISs for black-box discrete-time systems, where the dynamics are unknown and only limited simulation data are available. Traditionally, a CSIS requires that for every state in the set, there exists a control input that keeps the system within the set at the next step. However, enforcing such universal invariance, i.e., requiring the set to remain controlled invariant for all states, is often overly restrictive or impractical for black-box systems. To address this, we introduce the notion of a Probably Approximately Correct (PAC) CSIS, in which, with prescribed confidence, there exists a suitable control input to keep the system within the set at the next step for at least a specified fraction of the states. Our approach leverages barrier functions and scenario optimization, yielding a tractable linear programming method for estimating PAC CSISs. Several illustrative examples demonstrate the effectiveness of the proposed framework.
\end{abstract}


\end{frontmatter}

\section{Introduction}
\label{sec:intro}
Identifying a controlled safety invariant set (CSIS) is crucial in safety-critical applications, as it defines a region where the system can operate safely despite complex dynamics \citep{blanchini1999set,ames2019control}. Without control inputs, a safety invariant set (SIS) is a subset of the safe set containing states that remain within it at the next step. When control inputs are available, the SIS generalizes to a CSIS, comprising states that can be kept within the set using admissible controls. Knowing these sets is important for safety and controller design—for instance, a CSIS is essential for safe model predictive control \citep{kerrigan2000invariant}.

 The computation of CSISs has been extensively studied, with many model-based methods proposed \citep{sassi2012computation,schmid2023probabilistic,bravo2005computation,le2017interval}. However, these methods require knowledge of the system model, which is often unavailable for physical or black-box systems. Approximating the dynamics is typically difficult, costly, and can produce models too complex for practical CSIS computation. Moreover, the standard one-step invariance condition over all states in CSIS generally implies always-invariance, which is overly strong for black-box systems, where reliable predictions are limited to one or a few steps \citep{janner2019trust}. To address this, we decouple one-step invariance from always-invariance by relaxing the requirement of universal one-step invariance over all states in the CSIS to only a fraction of states, and adapt it to the black-box setting. Using only a finite dataset, we learn CSISs within the PAC framework \citep{valiant1984theory}, ensuring that, with specified confidence, a certain fraction of initial states in the learned CSIS can keep the system within the set at the next step.

In this paper, we propose a data-driven approach to estimate CSISs in the PAC sense for black-box discrete-time dynamical systems, without resorting to system modeling. We assume that dynamical models of these systems are unknown, and we have access to a numeric simulator or system behaviors sampled from the actual system. The method relies on barrier certificates and a family of data to construct linear optimization problems for computing CSISs, and utilizes the scenario optimization theory, originally designed for solving uncertain convex programs statistically in \citep{calafiore2006scenario}, to formally characterize the one-step invariance of the computed CSIS in the PAC sense. Finally, we test the proposed method on computing CSISs using several examples. The main contributions of this work are summarized as follows: \begin{enumerate} 
\item This work addresses the challenge of identifying CSISs for black-box discrete-time systems. 
\item A linear programming-based method is proposed to compute CSISs from a family of data. Consequently, we reduce the highly complex nonlinear problem of computing CSISs for black-box systems to a computationally tractable linear programming problem. 
\item A prototype tool \textbf{PCSIS} is developed to implement the proposed approach and is available at https://github.com/pcsis/PCSIS. The effectiveness of our approach is demonstrated through several numerical examples.
\end{enumerate}

\subsection*{Related Work}
Computing CSISs has been extensively studied, e.g., \citep{blanchini1999set,korda2013convex,bravo2005computation,xue2021robust}. However, these methods rely on having an explicit analytical model of the system dynamics, making them unsuitable for black-box systems where no prior model is available—the focus of this work.

To overcome the limits of model-based methods, recent work uses data-driven approaches. For linear systems, \citep{chen2018data} approximates robust invariant sets, and \citep{chakrabarty2018data} applies active learning to nonlinear systems, but neither provides formal guarantees. \citep{wang2020scenario} identifies a PAC set implicitly via fixed-point iteration and fits a classifier for uncontrolled systems, without guarantees. In contrast, our method computes explicit CSISs with PAC guarantees, while handling the added complexity of control inputs.

Scenario optimization, introduced in \citep{calafiore2006scenario}, addresses uncertain convex programs by sampling constraints, reducing them to simpler finite problems. It has been extended to estimate reachable sets \citep{devonport2020estimating,xue2020pac,dietrich2024nonconvex} and, when combined with barrier certificates, to safety verification and safe controller design \citep{salamati2022data,salamati2022safety,nejati2023formal,salamati2024data}. Existing barrier-based methods typically require Lipschitz constants, which are difficult to compute. In contrast, our CSIS computation avoids this requirement, relies mainly on linear programming, and focuses on maximizing CSIS size rather than just verification.

A related work, \citep{korda2020computing}, provides PAC guarantees for approximating maximal invariant sets, but its sample complexity depends on system dimension and hard-to-estimate quantities like Lipschitz constants. In contrast, we focus on PAC guarantees for one-step invariance, which leads to dimension-independent sample complexity and avoids relying on such intractable parameters.



\section{Preliminaries}
\label{sec:pre}

This section defines the problem of computing PAC CSISs for black-box discrete-time systems and recalls scenario optimization for uncertain convex problems.

\noindent\textbf{Notations.} $\mathbb{R}$, $\mathbb{R}_{\geq 0}$, and $\mathbb{R}^n$ denote real numbers, non-negative real numbers, and $n$-dimensional vectors, respectively; $\mathbb{N}$ denotes non-negative integers. For sets $\Delta_1$ and $\Delta_2$, $\Delta_1 \setminus \Delta_2$ is the set difference.

\subsection{Problem Statement}

We consider a black-box discrete-time system
\begin{equation}
\label{system_c}
\bm{x}(t+1) = \bm{f}(\bm{x}(t), \bm{u}(t)), \quad t \in \mathbb{N},
\end{equation}
with state $\bm{x}(\cdot) \in \mathbb{R}^n$, control $\bm{u}(\cdot) \in \mathbb{U} \subseteq \mathbb{R}^s$ (compact), and unknown continuous dynamics $\bm{f}$ over $\bm{x}$ and $\bm{u}$. A control policy $\pi$ is a function $\bm{u}(\cdot): \mathbb{N}\rightarrow \mathbb{U}$.

\begin{assumption}
\label{black_c_assumption}
We can extract a set of independent and identically distributed (i.i.d.) state samples \[\mathbb{X}=\{\bm{x}_i\}_{i=1}^N\sim \textnormal{P}^N,\] where $\textnormal{P}^N$ denotes the $N$-fold product distribution of $\textnormal{P}$, and each sample $\bm{x}_i$ is drawn independently from the uniform probability space $(\mathcal{X},\mathcal{F}_{\bm{x}},\textnormal{P})$. For any $\bm{x} \in \mathcal{X}$ and $\bm{u} \in \mathbb{U}$, the next state $\bm{y} = \bm{f}(\bm{x}, \bm{u})$ can be observed.
\end{assumption}

A CSIS is a set of states in $\mathcal{X}$ from which a control policy can keep the system within $\mathcal{X}$ at the next step.

\begin{definition}[One-Step CSIS]
\label{c_safe}
A set $\mathbb{S} \subseteq \mathcal{X}$ is a CSIS if, for each $\bm{x}_0 \in \mathbb{S}$, there exists $\bm{u} \in \mathbb{U}$ such that $\bm{f}(\bm{x}_0, \bm{u}) \in \mathbb{S}$, i.e., $\forall \bm{x}_0\in \mathbb{S}. \exists \bm{u}\in \mathbb{U}. \bm{f}(\bm{x}_0, \bm{u}) \in \mathbb{S}$.
\end{definition}

As noted in the introduction, enforcing one-step invariance for every state in $\mathbb{S}$ can be overly restrictive or impractical, particularly for black-box systems. This motivates adopting a probabilistic notion of one-step CSISs.

\begin{definition}[Probabilistic One-Step CSIS]
\label{pro_csis}
Given probability threshold $\alpha \in (0,1)$, a set $\tilde{\mathbb{S}} \subseteq \mathcal{X}$ with $\textnormal{P}[\bm{x} \in \tilde{\mathbb{S}}]\neq 0$ is a probabilistic one-step CSIS with respect to $\alpha$ if
\begin{equation}
\label{pone}
\textnormal{P}[\exists \bm{u} \in \mathbb{U}, \bm{f}(\bm{x},\bm{u}) \in \tilde{\mathbb{S}} \mid \bm{x} \in \tilde{\mathbb{S}}] 
\ge 1 - \frac{\alpha}{\textnormal{P}[\bm{x} \in \tilde{\mathbb{S}}]}.
\end{equation}
\end{definition}

\eqref{pone} means that, among all states in $\tilde{\mathbb{S}}$, the fraction of states 
for which there exists a control input that keeps the system \eqref{system_c} 
inside $\tilde{\mathbb{S}}$ at the next time step is at least
$1 - \frac{\alpha}{\textnormal{P}[\bm{x} \in \tilde{\mathbb{S}}]}$.

Since only sampled data are available for the system \eqref{system_c}, exact computation of $\tilde{\mathbb{S}}$ is generally impossible under Assumption \ref{black_c_assumption}. Thus, we aim to compute a PAC CSIS.

\begin{problem}[PAC CSIS]
\label{pac_csis}
Given $\mathcal{X}$, probability threshold $\alpha$, and confidence $\beta$, design a procedure $\Xi(\mathbb{X})$ that returns $\tilde{\mathbb{S}}$ with $\textnormal{P}[\bm{x} \in \tilde{\mathbb{S}}]\neq 0$ such that, with confidence at least $1-\beta$ over the sample set $\mathbb{X}$, 
\[
\textnormal{P}[\exists \bm{u} \in \mathbb{U}, \bm{f}(\bm{x},\bm{u}) \in \tilde{\mathbb{S}} \mid \bm{x} \in \tilde{\mathbb{S}}] 
\ge 1 - \frac{\alpha}{\textnormal{P}[\bm{x} \in \tilde{\mathbb{S}}]}.
\]
\end{problem}

If both $1-\beta$ and $1-\frac{\alpha}{\textnormal{P}[\bm{x} \in \tilde{\mathbb{S}}]}$ are high, the system starting in $\tilde{\mathbb{S}}$ remains in $\tilde{\mathbb{S}}$ at the next step with high likelihood. Since $\textnormal{P}[\bm{x}\in \tilde{\mathbb{S}}]$ is positively correlated with $1 - \frac{\alpha}{\textnormal{P}[\bm{x}\in \tilde{\mathbb{S}}]}$, maximizing $\textnormal{P}[\bm{x} \in \tilde{\mathbb{S}}]$ helps enlarge the PAC CSIS while maintaining safety guarantees.

In our framework, $\mathrm{P}$ is a user-specified reference measure on the state space $\mathcal{X}$, used to quantify set size and define the measure under which sampling and safety guarantees are evaluated. In particular, $\mathrm{P}[\bm{x} \in \tilde{\mathbb{S}}]$ denotes the relative volume of $\tilde{\mathbb{S}}$, and the conditional term in~\eqref{pone} captures the fraction of states in $\tilde{\mathbb{S}}$ that admit a control input keeping the system within $\tilde{\mathbb{S}}$ at the next step. Importantly, $\mathrm{P}$ is not the true state distribution, but a fixed reference measure. Consequently, the one-step PAC guarantee-quantified by the reference volume $1 - \frac{\alpha}{\mathrm{P}[\bm{x}\in \tilde{\mathbb{S}}]}$ with confidence $1-\beta$-can be applied recursively whenever the system lies in $\tilde{\mathbb{S}}$, enabling safety monitoring over extended horizons.

\begin{remark}
Exact computation of $\textnormal{P}[\bm{x} \in \tilde{\mathbb{S}}]$ may be difficult for complex sets. Practical approaches include computing a lower bound or using Monte Carlo estimation.
\end{remark}

We observe a positive correlation between $\textnormal{P}[\bm{x}\in \tilde{\mathbb{S}}]$ and $1 - \frac{\alpha}{\textnormal{P}[\bm{x}\in \tilde{\mathbb{S}}]}$. Therefore, maximizing $\textnormal{P}[\bm{x} \in \tilde{\mathbb{S}}]$ helps enlarge the PAC CSIS while maintaining safety guarantees.

\subsection{Scenario Optimization}
\label{sub:so}

This section briefly summarizes scenario optimization, mainly following \citep{calafiore2006scenario}. Scenario optimization seeks robust solutions to uncertain convex problems of the form:
\begin{equation}
\label{ROP}
    \begin{split}
        \bm{z}^* &= \arg\min_{\bm{z}\in \mathbb{D}\subset \mathbb{R}^m} \bm{c}^{\top} \bm{z}, \\
        \text{s.t. } & \bm{g}(\bm{z},\bm{\delta})\leq 0, \forall \bm{\delta}\in \Delta,
    \end{split}
\end{equation}
where \(\bm{\delta}\in \Delta\) is an uncertain parameter with probability measure \(\textnormal{P}_\delta\). Finding a solution \(\bm{z}^*\) that satisfies all \(\bm{\delta} \in \Delta\) is generally infeasible.  

\begin{assumption}
\label{convex_ass}
The set \(\mathbb{D}\subseteq \mathbb{R}^m\) is convex and closed, \(\Delta \subseteq \mathbb{R}^{n_\delta}\), and \(\bm{g}(\bm{z},\bm{\delta})\) is continuous and convex in \(\bm{z}\) for any fixed \(\bm{\delta}\in \Delta\).
\end{assumption}
Scenario optimization approximates \eqref{ROP} using \(N\) samples of \(\bm{\delta}\), denoted $\Delta_N=\{\bm{\delta}_i\}_{i=1}^N\stackrel{\text{i.i.d.}}{\sim}\textnormal{P}_{\delta}$, called scenarios, leading to the sampled program:
\begin{equation}
\label{ROP-N}
    \begin{split}
        \bm{z}^*(\Delta_N) &= \arg\min_{\bm{z}\in \mathbb{D}\subset \mathbb{R}^m} \bm{c}^{\top} \bm{z}, \\
        \text{s.t. } & \bm{g}(\bm{z},\bm{\delta}_i)\leq 0, \quad i=1,\ldots,N.
    \end{split}
\end{equation}

\begin{assumption}
\label{feasible_uni}
For any sample set \(\Delta_N\), \eqref{ROP-N} has a unique optimal solution \(\bm{z}^*(\Delta_N)\) (tie-breaks applied if needed).
\end{assumption}
Once \(\bm{z}^*(\Delta_N)\) is obtained, we can define the set of constraints for which it is not satisfied: $\mathbb{F}(\bm{z}^*(\Delta_N)) = \{\bm{\delta} \in \Delta \mid \bm{g}(\bm{z}^*(\Delta_N), \bm{\delta}) \not\le 0\}$.
\begin{definition}
\label{violation}
The violation probability of \(\bm{z}^*(\Delta_N)\) is:  $
V(\bm{z}^*(\Delta_N)) = \textnormal{P}_\delta[\bm{\delta}\in \mathbb{F}(\bm{z}^*(\Delta_N))]$,
i.e., the probability that \(\bm{z}^*(\Delta_N)\) fails for a randomly drawn \(\bm{\delta}\).
\end{definition}

The main PAC guarantee is given by \citep{campi2009scenario}:

\begin{proposition}[\citep{campi2009scenario}]
\label{formal_PAC}
With \(N\) samples, 
\[
\textnormal{P}_{\Delta_N} [V(\bm{z}^*(\Delta_N)) \le \alpha] \ge 1-\beta,
\]
where $\textnormal{P}_{\Delta_N}:=\textnormal{P}_{\delta}^N$, provided \(\alpha \ge \frac{2}{N} (\ln\frac{1}{\beta} + m)\), and \(\textnormal{P}_\delta^N\) is the $N$-fold product distribution of $\textnormal{P}_{\delta}$.
\end{proposition}

\section{Estimating PAC CSISs}
\label{sec:methods}
In this section, we propose a linear programming method to compute a large PAC CSIS, based on scenario optimization and control barrier functions.
To handle the difficulty of control inputs, we first compute a PAC one-step safety invariant set (SIS) for an uncontrolled black-box system ($\bm{x}(t+1)=\bm{f}(\bm{x}(t))$), then extend the method to controlled systems. This helps illustrate how the key ideas interact.

\subsection{Computing PAC SISs}
\label{sec:pacsis}
In this subsection, we present our method to compute a PAC SIS for an uncontrolled black-box system. First, we recall discrete-time barrier functions for which the zero superlevel set defines a SIS.

\begin{definition}
   Given $\gamma\in (0,1)$, a function $h(\cdot): \mathbb{R}^n \rightarrow \mathbb{R}$ is a discrete-time barrier function if the following inequalities hold  with $\lambda=0$,
\begin{equation}
\label{barrier0}
    \begin{cases}
        h(\bm{x})<0, & \forall \bm{x}\in \mathbb{R}^n \setminus \mathcal{X},\\
        h(\bm{f}(\bm{x}))\geq \gamma h(\bm{x})-\lambda, & \forall \bm{x}\in \mathcal{X},\\
        \lambda\in [0,\overline{\lambda}],
    \end{cases}
\end{equation}
where $\overline{\lambda}$ is a specified upper bound for $\lambda$. 
\end{definition}
\begin{proposition} [\citep{agrawal2017discrete}]
    If $h(\cdot): \mathbb{R}^n \rightarrow \mathbb{R}$ is a discrete-time barrier function, the set $\widetilde{\mathbb{I}}=\{\bm{x}\in \mathcal{X}\mid h(\bm{x})\geq 0\}$ is a SIS, i.e., $\forall \bm{x}\in \widetilde{\mathbb{I}}, \bm{f}(\bm{x})\in \widetilde{\mathbb{I}}$.  
\end{proposition}

Traditionally, the computation of a discrete-time barrier function satisfying \eqref{barrier0} needs a parameterized function $h(\bm{a},\cdot): \mathbb{R}^n\rightarrow \mathbb{R}$. In this paper, we impose certain assumptions on $h(\bm{a},\cdot): \mathbb{R}^n\rightarrow \mathbb{R}$. 
\begin{assumption}
\label{h}
 $h(\bm{a},\cdot): \mathbb{R}^n\rightarrow \mathbb{R}$ is linear over unknown parameters $\bm{a} \in \mathbb{R}^m$ and negative over $\bm{x}\in \mathbb{R}^n\setminus \mathcal{X}$. Also, it is continuous over $\bm{a}\in \mathbb{R}^{m}$ for any $\bm{x}\in \mathbb{R}^n$. 
\end{assumption}

One example of a function satisfying Assumption \ref{h} is:
\begin{equation}
\label{h_type}
\begin{cases}
h(\bm{a},\bm{x}) = 1_{\mathcal{X}}(\bm{x}) h_1(\bm{a},\bm{x}) + 1_{\mathbb{R}^n \setminus \mathcal{X}}(\bm{x}) C, \\
h(\bm{a},\bm{x}) \equiv a_1 \quad (\text{if } a_2 = \dots = a_m = 0), \forall \bm{x} \in \mathcal{X},
\end{cases}
\end{equation}
where $1_{\mathcal{X}}$ is the indicator function, $h_1(\bm{a},\bm{x})$ is linear in $\bm{a}$ and continuous in both $\bm{a}$ and $\bm{x}$, and $C$ is a negative constant. For $\bm{x} \notin \mathcal{X}$, $h(\bm{a},\bm{x}) = C$. Using this type of barrier function not only simplifies the computations, but also  ensures the system remains safe without entering the unsafe set to observe its exact state, which is important in safety-critical applications.

We search for a parameterized function $h(\bm{a},\cdot): \mathbb{R}^n\rightarrow \mathbb{R}$ satisfying \eqref{barrier0} with $\mathbb{X}=\{\bm{x}_i\}_{i=1}^N$ under Assumption \ref{black_c_assumption} via solving the following linear program:
\begin{equation}
\label{one_N}
  \begin{split}
  &\textstyle \min_{\bm{a}\in \mathbb{R}^m,\lambda\in \mathbb{R}} \lambda\\
    \text{s.t.}&
    \begin{cases}
        h(\bm{a},\bm{x}_0)\geq \epsilon_0,\\
        h(\bm{a},\bm{y}_{i})\geq \gamma h(\bm{a},\bm{x}_i)-\lambda, i=1,\ldots,N,\\
         \lambda \in [0,\overline{\lambda}];~ a_l \in [-U_{a_l}, U_{a_l}];~
         l=1,\ldots, m,\\
    \end{cases}
    \end{split}
\end{equation}
where $\bm{y}_i= \bm{f}(\bm{x}_i)$ for $i=1,\ldots,N$. The parameters $\epsilon_0>0$ and $\gamma\in(0,1)$ are user-defined, and $U_{a_l} \ge \epsilon_0$ bounds $|a_l|$. The value $\overline{\lambda}$ just needs to exceed $\gamma U_{a_1}-C$ (by setting $a_2=\ldots=a_m=0$ and $a_1\geq 0$), ensuring that the optimization \eqref{one_N} is feasible for any set of $N$ samples $\mathbb{X}$, as required by Assumption \ref{feasible_uni}. \textit{The requirement that $h(\bm{a},\bm{x}_0)\geq \epsilon_0$ is to guarantee, due to the continuity of the function $h(\bm{a},\bm{x})$ over $\mathcal{X}$, that the set $\widetilde{\mathbb{S}}=\{\bm{x}\in \mathcal{X}\mid h(\bm{a},\bm{x}) \geq 0\}$ has positive Lebesgue measure, i.e.,  $\textnormal{P}[\bm{x} \in \tilde{\mathbb{S}}]\neq 0$. Generally, $\bm{x}_0\in \mathcal{X}$ is recommended to be a state that is away from the boundary of the set $\mathcal{X}$.}

Let $(\bm{a}^{*}(\mathbb{X}),\lambda^*(\mathbb{X}))$ be an optimal solution to \eqref{one_N}. If $\lambda^*(\mathbb{X})=0$, we proceed with addressing the following optimization \eqref{barrier_convex1} that maximizes the size of the set $\widetilde{\mathbb{S}}=\{\bm{x}\in \mathcal{X}\mid h(\bm{a},\bm{x})\geq 0\}$ according to the tie-break rule as suggested in \citep{calafiore2006scenario}:
\begin{equation}
    \label{barrier_convex1}
  \begin{split}
    &\textstyle\max_{\bm{a}\in \mathbb{R}^m } \sum_{i=1}^{N'}h(\bm{a},\bm{x}'_i)\\
    \text{s.t.}&
    \begin{cases}
           { h(\bm{a},\bm{x}_0)\geq \epsilon_0,}\\
        h(\bm{a},\bm{f}(\bm{x}_i))\geq \gamma h(\bm{a},\bm{x}_i), i=1,\ldots,N, \\
       a_l \in [-U_{a_l}, U_{a_l}], l=1,\ldots,m, \\
    \end{cases}
    \end{split}
\end{equation}
where \(\{\bm{x}'_i\}_{i=1}^{N'}\) is a set of states evenly sampled from the safe set \(\mathcal{X}\). 
\textbf{This dataset is different from $\mathbb{X}$. It is fixed in the whole computations.} 
Ideally, the linear program would maximize the volume of 
\(\widetilde{\mathbb{S}} = \{\bm{x} \in \mathcal{X} \mid h(\bm{a}, \bm{x}) \ge 0\}\), i.e., 
\(\max_{\bm{a}\in \mathbb{R}^m} \texttt{VOL}(\widetilde{\mathbb{S}})\), 
but computing this directly is intractable. 
Instead, we approximate it by maximizing the surrogate 
\(\sum_{i=1}^{N'} h(\bm{a}, \bm{x}'_i)\), which is linear in \(\bm{a}\) and computationally efficient.
 
Following Proposition \ref{formal_PAC}, we have the conclusion below. 
\begin{proposition}
\label{pro3}
   Let $\alpha, \beta\in (0,1)$, $(\bm{a}^{*}(\mathbb{X}),\lambda^*(\mathbb{X}))$ be the optimal solution to \eqref{one_N} (if $\lambda^*(\mathbb{X})=0$, $\bm{a}^{*}(\mathbb{X})$ is obtained via solving \eqref{barrier_convex1} according to the tie-break rule) with $\mathbb{X}=\{\bm{x}_i\}_{i=1}^m \stackrel{\text{i.i.d.}}{\sim}\textnormal{P}$, and define the computational procedure
   \begin{equation}
\label{cp}
\Xi(\mathbb{X}) \triangleq 
\begin{cases}
1, & \text{if } \lambda^*(\mathbb{X})=0,\\
0, & \text{otherwise}.
\end{cases}
\end{equation}
Then, if $N$ and $m$ satisfies $\alpha\geq {\frac{2}{N}(\ln{\frac{1}{\beta}}+m+1)}$,
the following guarantee holds:
\begin{equation*}
    \begin{split}
 &   \textnormal{P}_{\mathbb{X}}\left[ 
 \begin{split}
 \Xi(\mathbb{X})\Rightarrow  
 \begin{split}
 &\textstyle\textnormal{P}[\bm{f}(\bm{x})\in \widetilde{\mathbb{S}}\mid \bm{x}\in \widetilde{\mathbb{S}}]\\
 &\geq 1-\frac{1}{\textnormal{P}[\bm{x}\in\widetilde{\mathbb{S}}]}\alpha
 \end{split}
 \end{split}
 \right]\geq 1-\beta,
    \end{split}
\end{equation*}
where $\widetilde{\mathbb{S}}=\{\bm{x}\in \mathcal{X}\mid h(\bm{a}^{*}(\mathbb{X}),\bm{x})\geq 0\}$ and $\textnormal{P}_{\mathbb{X}}:=\textnormal{P}^N$.  
\end{proposition}
\begin{pf}
The optimization \eqref{one_N} is a scenario optimization of the following uncertain convex optimization, which satisfies Assumption \ref{convex_ass},  
\begin{equation*}
    \label{barrier_convex0}
  \begin{split}
    &\textstyle\min_{\bm{a}\in \mathbb{R}^m,\lambda\in \mathbb{R}} \lambda\\
    \text{s.t.}&
    \begin{cases}
       { h(\bm{a},\bm{x}_0)\geq \epsilon_0,}\\
        h(\bm{a},\bm{f}(\bm{x}))\geq \gamma h(\bm{a},\bm{x}){-\lambda}, \forall \bm{x}\in \mathcal{X}, \\
        \lambda \in [0,\overline{\lambda}]; ~a_l \in [-U_{a_l};~ U_{a_l}];~ l=1,\ldots,m. \\
    \end{cases}
    \end{split}
\end{equation*}
Since $(\bm{a}^*(\mathbb{X}),\lambda^*(\mathbb{X}))$ is an optimal solution to \eqref{one_N},  according to  Proposition \ref{formal_PAC}, we have 
$\textnormal{P}_{\mathbb{X}}[\textnormal{P}[h(\bm{a}^*(\mathbb{X}),\bm{f(\bm{x})})< \gamma h(\bm{a}^*(\mathbb{X}),\bm{x})-\lambda^*(\mathbb{X})]\leq \alpha] \geq 1-\beta$. Also, since the uniform distribution is assigned to the set $\mathcal{X}$, then, 
$\textnormal{P}_{\mathbb{X}}[\lambda^*(\mathbb{X})=0 \Rightarrow \textnormal{P}[h(\bm{a}^*(\mathbb{X}),\bm{f(\bm{x})})\geq \gamma h(\bm{a}^*(\mathbb{X}),\bm{x})]\geq 1-\frac{1}{\textnormal{P}[\bm{x}\in\widetilde{\mathbb{S}}]}\alpha] \geq 1-\beta$, which can be obtained as follows:
\begin{equation*}
\label{conditional}
    \begin{split}
     &\textnormal{P}[h(\bm{a}^*(\mathbb{X}),\bm{f(\bm{x})})\geq \gamma h(\bm{a}^*(\mathbb{X}),\bm{x})\mid \bm{x}\in \widetilde{\mathbb{S}}]\\
     =&1-\textnormal{P}[h(\bm{a}^*(\mathbb{X}),\bm{f(\bm{x})})< \gamma h(\bm{a}^*(\mathbb{X}),\bm{x})\mid \bm{x}\in \widetilde{\mathbb{S}}]\\
     =&1-\frac{\textnormal{P}[\bm{x} \in \widetilde{\mathbb{S}}\wedge h(\bm{a}^*(\mathbb{X}),\bm{f(\bm{x})})< \gamma h(\bm{a}^*(\mathbb{X}),\bm{x})]}{\textnormal{P}[\bm{x}\in\widetilde{\mathbb{S}}]}\\
     \geq & \textstyle1-\frac{1}{\textnormal{P}[\bm{x}\in\widetilde{\mathbb{S}}]}\alpha.
    \end{split}
\end{equation*}

Further, if $\bm{x}\in \widetilde{\mathbb{S}}$ and $h(\bm{a}^*(\mathbb{X}),\bm{f(\bm{x})})\geq \gamma h(\bm{a}^*(\mathbb{X}),\bm{x})  \geq 0$ holds, then $\bm{f}(\bm{x}) \in \widetilde{\mathbb{S}}$, which implies 
\[
\begin{split}
&\textnormal{P}[\bm{f}(\bm{x})\in \widetilde{\mathbb{S}}\mid  \bm{x}\in \widetilde{\mathbb{S}}]\\
&\geq\textnormal{P}[h(\bm{a}^*(\mathbb{X}),\bm{f(\bm{x})})\geq \gamma h(\bm{a}^*(\mathbb{X}),\bm{x})\mid \bm{x}\in \widetilde{\mathbb{S}}].
\end{split}\]
Consequently, the conclusion holds. \qed
\end{pf}

\subsection{Computing PAC CSISs}
\label{sec:paccsis}
In this section, we present our method for computing PAC CSISs for systems with control inputs, using scenario optimization and discrete-time control barrier functions.

We first review discrete-time control barrier functions, whose zero-superlevel set defines a CSIS, and highlight the added challenges compared to uncontrolled systems.

\begin{definition}
   Given $\gamma\in (0,1)$, a continuous function $h(\cdot): \mathbb{R}^n \rightarrow \mathbb{R}$ is a discrete-time control barrier function if the following inequalities hold with $\lambda=0$:
\begin{equation}
    \begin{cases}
        h(\bm{x})<0, & \forall \bm{x}\in \mathbb{R}^n \setminus \mathcal{X},\\
        \max_{\bm{u}\in \mathbb{U}}h(\bm{f}(\bm{x},\bm{u}))\geq \gamma h(\bm{x})-\lambda, & \forall \bm{x}\in \mathcal{X},\\
       {\lambda \in [0,\overline{\lambda}].}
    \end{cases}
\end{equation}
\end{definition}

\begin{proposition}[\citep{agrawal2017discrete}]
    If $h(\cdot): \mathbb{R}^n \rightarrow \mathbb{R}$ is a discrete-time control barrier function, the set $\{\bm{x}\in \mathcal{X}\mid h(\bm{x})\geq 0\}\neq \emptyset$ is a CSIS.      
\end{proposition}

The computational challenge associated with discrete-time control barrier functions stems from $\max_{\bm{u}\in \mathbb{U}}h(\bm{f}(\bm{x},\bm{u}))\geq \gamma h(\bm{x})-\lambda, \forall \bm{x}\in \mathcal{X}$, where the term $\max_{\bm{u}\in \mathbb{U}}h(\bm{f}(\bm{x},\bm{u}))$ introduces significant computational complexity.  A strategy will help us to remove the $\max$ operator and facilitate the computations. It is employing the empirical average: 
\begin{equation}
\label{samples}
\sum_{j=1}^M w_j(\bm{x}) h(\bm{f}(\bm{x},\bm{u}_j))\geq \gamma h(\bm{x})-\lambda, \forall \bm{x}\in \mathcal{X},
\end{equation}
where $\{\bm{u}_{j}\}_{j=1}^M$ are the grid points from the uniform discretization of the set $\mathbb{U}$, and $w_{j}(\cdot): \mathbb{R}^n\rightarrow \mathbb{R}_{\geq 0}$ with $\sum_{j=1}^M w_{j}(\bm{x})=1$. Since  $\sum_{j=1}^M w_j(\bm{x}) h(\bm{f}(\bm{x},\bm{u}_j)) \leq \max_{\bm{u}\in \mathbb{U}}h(\bm{f}(\bm{x},\bm{u}))$ for $\bm{x}\in \mathcal{X}$, we have that \eqref{samples} implies $\max_{\bm{u}\in \mathbb{U}}h(\bm{f}(\bm{x},\bm{u}))\geq \gamma h(\bm{x})-\lambda, \forall \bm{x}\in \mathcal{X}$.

When using \eqref{samples} for computations, we need to determine the function $w_{j}(\cdot): \mathbb{R}^n\rightarrow \mathbb{R}_{\geq 0}$ appropriately. However, the function $\bm{f}(\cdot,\cdot): \mathbb{R}^n\times \mathbb{R}^s \rightarrow \mathbb{R}^n$ is unknown, this determination is nontrivial. To address this, we propose an $\epsilon$-greedy iteration algorithm to learn these functions from a finite dataset
$X^L=\{\big(\bm{x}_i^L,\bm{u}_{j},\bm{y}_{i,j}^L\big)\}_{i=1,j=1}^{N,M}$ (\textbf{which is distinct from $\mathbb{X}$ and is only used to select the functions $w_{j}(\cdot): \mathbb{R}^n \rightarrow \mathbb{R}_{\ge 0}$, and is therefore fixed in our computations}), with $\bm{x}_i^L \in \mathcal{X}(i=1,\ldots,N)$ and $\bm{y}_{i,j}^L=\bm{f}(\bm{x}_i^L,\bm{u}_{j})$ observed by simulating the system \eqref{system_c} with the initial state $\bm{x}_i^L$ and control input $\bm{u}_{j}$. 

Based on the dataset $X^L$, the constraint \eqref{samples} is relaxed to 
\begin{equation}
\label{average}
\sum_{j=1}^M w_{j}^L(\bm{x}_i^L) h(\bm{y}_{i,j}^L)\geq \gamma h(\bm{x}_i^L)-\lambda, i=1,\ldots,N,
\end{equation}
where $w_{j}^L(\cdot): \mathbb{R}^n\rightarrow \mathbb{R}_{\geq 0}$ with $\sum_{j=1}^M w_{j}^L(\bm{x})=1$. Using a parameterized function $h(\bm{a},\cdot): \mathbb{R}^n\rightarrow \mathbb{R}$ satisfying Assumption \ref{h}, we first construct a linear program to compute $\bm{a}$ below:
\begin{equation}
\label{N_N_u0}
  \begin{split}
    &\textstyle\min_{\lambda \in \mathbb{R}, \bm{a}=(a_1,\ldots,a_m)^{\top}\in \mathbb{R}^m } \lambda\\
    \text{s.t.}&
    \begin{cases}
    { h(\bm{a},\bm{x}_0)\geq \epsilon_0,}\\
        \sum_{j=1}^M w_{j}^L(\bm{x}_i^L) h(\bm{a},\bm{y}_{i,j}^L)\geq \gamma h(\bm{a},\bm{x}_i^L)-\lambda, \\
        \lambda \in [0, \overline{\lambda}], a_l \in [-U_{a_l}, U_{a_l}], \\
        i=1,\ldots,N; l=1,\ldots, m, 
    \end{cases}
    \end{split}
\end{equation}
where $w_{j}^L(\bm{x})=\frac{1}{M}$ for $j=1,\ldots,M$, 
and $\gamma\in (0,1)$, $\overline{\lambda}>0$, and $U_{a_l}$ have the same meaning as before, $l=1,\ldots,m$. 

In \eqref{N_N_u0}, we give all sampled control inputs equal weight $\frac{1}{M}$ instead of choosing one randomly, which reduces variance. Let \((\bm{a}_0^*, \lambda_0^*)\) be the optimal solution. Then, we have $\max_{\bm{u}\in \mathbb{U}}h(\bm{a}^*_0,\bm{f}(\bm{x},\bm{u})) \geq \sum_{j=1}^M w_{j}(\bm{x}_i^L) h(\bm{a}_0^*,\bm{y}_{i,j}^L)\geq \gamma h(\bm{a}^*_0,\bm{x}_i^L)-\lambda_0^*$.

    If $\lambda_0^*\neq 0$, we further update the barrier function to refine $\lambda_0^*$. Without loss of generality, we assume $\bm{u}_i=\arg\max_{\bm{u}_j,j=1,\ldots,M} h(\bm{a}_0^{*},\bm{f}(\bm{x}_i^L,\bm{u_j}))$, $i=1,\ldots,N$.  It is observed that, for $i=1,\ldots,N$, 
\begin{equation}
\label{epsilon}
    (1-\epsilon)h(\bm{a}_0^*,\bm{y}_{i,i}^L)+ \sum_{j=1}^M \frac{\epsilon}{M} h(\bm{a}_0^*,\bm{y}_{i,j}^L)\geq \sum_{j=1}^M \frac{1}{M} h(\bm{a}_0^*,\bm{y}_{i,j}^L),
\end{equation}
where $\epsilon \in (0,1]$ is a user specified value. Therefore, we next construct a linear program to update $\bm{a}$ in the parameterized barrier function $h(\bm{a},\bm{x})$ and $\lambda$ as follows:
\begin{equation}
\label{N_N_u2}
  \begin{split}
    &\textstyle\min_{\lambda \in \mathbb{R}, \bm{a}\in \mathbb{R}^m } \lambda\\
    \text{s.t.}&
    \begin{cases}
          { h(\bm{a},\bm{x}_0)\geq \epsilon_0,}\\
        \sum_{j=1}^M w_{j}^L(\bm{x}_i^L) h(\bm{a},\bm{y}_{i,j}^L)\geq \gamma h(\bm{a},\bm{x}_i^L)-\lambda, \\
        \lambda \in [0, \overline{\lambda}], a_l \in [-U_{a_l}, U_{a_l}], \\
        i=1,\ldots,N; l=1,\ldots, m, 
    \end{cases}
    \end{split}
\end{equation}
where $w_j^L(\bm{x})=$
\begin{equation}
\label{weight0}
\begin{cases}
    1-\epsilon+\frac{\epsilon}{M}, & \text{if~}\bm{u}_j=\argmax\limits_{{\{\bm{u}_t\}_{t=1}^M}} h(\bm{a}_0^{*},\bm{f}(\bm{x},\bm{u}_t)),\\
    \frac{\epsilon}{M},  &\text{otherwise}
\end{cases}
\end{equation}
for $j=1,\ldots,M$, $\gamma\in (0,1)$, $\overline{\lambda}>0$, and $U_{a_l}$ have the same meaning as before, $l=1,\ldots,m$.  

\begin{proposition}
\label{xi_de}
    Let $(\lambda_1^{*},\bm{a}_1^*)$ be the optimal solution to \eqref{N_N_u2}. $\lambda_1^{*} \leq \lambda_0^*$ holds.
\end{proposition}
\begin{pf}
    \eqref{epsilon} implies that $(\lambda_0^*,\bm{a}_0^*)$ is a feasible solution to \eqref{N_N_u2}. Thus, $\lambda_1^{*} \leq \lambda_0^*$.  \qed
\end{pf}

After solving the optimization \eqref{N_N_u2} to find the optimal solution $(\lambda_1^*, \bm{a}_1^*)$, if $\lambda_1^* \neq 0$, we update the weight function $w_j^L(\cdot): \mathbb{R}^n \rightarrow \mathbb{R}$ using the $\epsilon$-greedy strategy \eqref{weight0}. In this update, we replace $h(\bm{a}_0^*,\bm{x})$ with the newly computed function $h(\bm{a}_1^*,\bm{x})$. We then solve a linear program of the form \eqref{N_N_u2} to further refine $\lambda_1^*$. We assume $\bm{u}_{i_j}=\arg\max_{\bm{u}_j,j=1,\ldots,M} h(\bm{a}_1^{*},\bm{f}(\bm{x}_i^L,\bm{u}_j))$, $i=1,\ldots,N$.  It is observed that, for $i=1,\ldots,N$, 
\begin{equation}
\label{epsilon2}
\begin{split}
    &\textstyle(1-\epsilon)h(\bm{a}_1^*,\bm{y}_{i,i_j}^L)+ \sum_{j=1}^M \frac{\epsilon}{M} h(\bm{a}_1^*,\bm{y}_{i,j}^L) \\
    \geq & (1-\epsilon)h(\bm{a}_1^*,\bm{y}_{i,i}^L)+ \sum_{j=1}^M \frac{\epsilon}{M} h(\bm{a}_1^*,\bm{y}_{i,j}^L),
\end{split}
\end{equation}
where $\epsilon \in (0,1]$ is the user-specified value in \eqref{epsilon}. Thus, $(\lambda_1^*,\bm{a}_1^*)$ is also a feasible solution to the new linear program. Let $(\lambda_2^*,\bm{a}_2^*)$ be the new optimal solution. Thus, we can ensure $\lambda_2^* \leq \lambda_1^*$. This iterative process is repeated until either the computed optimal value $\lambda$ is zero or a specified maximum number of iterations is reached.

If the maximum number of iterations is reached with \(\lambda \neq 0\), the process stops. 
Otherwise, at iteration \(k\), let \((\lambda_k^*, \bm{a}_k^*)\) be the solution with \(\lambda_k^* = 0\). 
We then update the weights \(w_j^L(\cdot)\) so that 
\(\sum_{j=1}^M w_j^L(\bm{x}) = 1\), and iteratively expand the set 
\(\{\bm{x}\in \mathcal{X} \mid h(\bm{a},\bm{x})\ge 0\}\) 
by solving a linear program at each step. 
The first iteration of solving a linear program \eqref{N_N_u3} is described below; later iterations follow the same procedure with \(\epsilon\)-greedy updates.
\begin{equation}
\label{N_N_u3}
  \begin{split}
    &\textstyle\max_{\bm{a}\in \mathbb{R}^m }  \sum_{i=1}^{N'}h(\bm{a},\bm{x}'_i)\\
    \text{s.t.}&
    \begin{cases}
        {h(\bm{a},\bm{x}_0)\geq \epsilon_0,}\\
        \sum_{j=1}^M w^L_{j}(\bm{x}_i^L) h(\bm{a},\bm{y}_{i,j}^L)\geq \gamma h(\bm{a},\bm{x}_i^L), \\
        a_l \in [-U_{a_l}, U_{a_l}]; i=1,\ldots,N;l=1,\ldots, m, 
    \end{cases}
    \end{split}
\end{equation}
where $w_j^L(\bm{x})$=
\begin{equation}
\label{weight}
\begin{cases}
    1-\epsilon+\frac{\epsilon}{M},&\text{if~}\bm{u}_j=\argmax\limits_{\{\bm{u}_t\}_{t=1}^M} h(\bm{a}_{k}^{*},\bm{f}(\bm{x},\bm{u}_t)),\\
    \frac{\epsilon}{M},&\text{otherwise}.
\end{cases}
\end{equation}
for $j=1,\ldots,M$, and $\gamma\in (0,1)$, $\overline{\lambda}>0$, and $U_{a_l}$ have the same meaning as before, $l=1,\ldots,m$.  \textbf{Similarly, \textit{$\{\bm{x}'_i\}_{i=1}^{N'}$ is fixed through the whole computations, as the one in \eqref{barrier_convex1}.} } 

The iterative process of updating the functions $w_j^L(\cdot) : \mathbb{R}^n \rightarrow \mathbb{R}, j = 1, \ldots, M$, which helps to enlarge the set
$\widetilde{\mathbb{S}}^L = \{\bm{x} \in \mathcal{X} \mid h(\bm{a}, \bm{x}) \ge 0\}$, stops when either the change in the objective value between iterations,  
$\Big|\sum_{i=1}^{N'}h(\bm{a}^*_k,\bm{x}'_i)-\sum_{i=1}^{N'}h(\bm{a}^*_{k-1},\bm{x}'_i)\Big|$, is below a threshold, or a maximum number of iterations is reached.
If the objective is replaced by 
\(\max_{\bm{a}\in \mathbb{R}^m} \texttt{VOL}(\widetilde{\mathbb{S}}^L)\), 
the set volume is guaranteed not to decrease at each step 
(proof similar to Proposition \ref{xi_de}, omitted). This process is summarized in Alg. \ref{alg1}. 
\begin{algorithm}[t]
\caption{Learning $w^L_j(\cdot): \mathbb{R}^n\rightarrow \mathbb{R}_{\geq 0}, j=1,\ldots,M$.}
    \begin{algorithmic}[1]
    \Require a probability error $\alpha\in [0,1)$; a confidence level $\beta \in [0,1)$; a parameterized barrier function $h(\bm{a},\cdot): \mathbb{R}^n\rightarrow \mathbb{R}$ of type \eqref{h_type} with $\bm{a}\in \mathbb{R}^m$; {$\bm{x}_0\in \mathcal{X}$; $\epsilon_0>0$; $\overline{\lambda}>0$;} $\gamma \in (0,1)$; $\epsilon \in (0,1)$; the smallest integer $N$ satisfying $\alpha\geq \frac{2}{N}(\ln\frac{1}{\beta}+m+1)$;
    datasets $X'=\{\bm{x}'_i\}_{i=1}^{N'}$ and $X^L=\{\big(\bm{x}_i^L,\bm{u}_{j},\bm{y}_{i,j}^L\big)\}_{i=1,j=1}^{N,M}$; termination threshold $\epsilon'>0$; maximum iteration numbers $K$ and $K'$.
    \Ensure  $w_j^L(\cdot):\mathbb{R}^n \rightarrow \mathbb{R}_{\geq 0}, j=1,\ldots,M$. 
    \State Initialize $w^L_j(\bm{x})=\frac{1}{M}$ for $\bm{x}\in \mathcal{X}$, $j=1,\ldots,M$;
     \State Solve the optimization \eqref{N_N_u0} to obtain $\lambda_0^*$ and $\bm{a}_0^*$;
    \State $k:=0$;
    \While{$\lambda_k^*>0$ and $k< K$}
    \State $k:=k+1$;
    \State Update $w_j^L(\bm{x})$ according to \eqref{weight0}, replacing  $h(\bm{a}_{k-1}^{*},\bm{f}(\bm{x},\bm{u}))$ with $h(\bm{a}_{k}^{*},\bm{f}(\bm{x},\bm{u}))$; 
    \State Solve the optimization \eqref{N_N_u2} to obtain $\lambda_k^*$ and $\bm{a}_k^*$;
    \EndWhile
    \If{$k\geq K$ and $\lambda_k^*>0$}
      \Return "fail". 
     \EndIf
    \While{$k< K'$}
    \State $k:=k+1$;
    \State Update $w_j^L(\bm{x})$ according to \eqref{weight}, replacing $h(\bm{a}_{k-1}^{*},\bm{f}(\bm{x},\bm{u}))$ with $h(\bm{a}_{k}^{*},\bm{f}(\bm{x},\bm{u}))$;
    \State Solve the optimization \eqref{N_N_u3} to obtain $\bm{a}_k^*$;
    \If{$|\sum_{i=1}^{N'} h(\bm{a}^*_k,\bm{x}'_i) - \sum_{i=1}^{N'} h(\bm{a}^*_{k-1},\bm{x}'_i)|\leq \epsilon'$}
    \State \textbf{break}
    \EndIf
    \EndWhile \\
    \Return $w_j^L(\cdot):\mathbb{R}^n \rightarrow \mathbb{R}_{\geq 0}, j=1,\ldots,M$.
    \end{algorithmic}
    \label{alg1}
\end{algorithm}

Finally, we solve a linear program, which is constructed from the sample set $\mathbb{X}=\{\bm{x}_i\}_{i=1}^N$ and the \textit{fixed} functions $w_j^L(\cdot):\mathbb{R}^n \rightarrow \mathbb{R}, j=1,\ldots,M$ returned by Alg. \ref{alg1}:
\begin{equation}
\label{N_N_u_r_l1}
  \begin{split}
    &\textstyle\min_{\bm{a}\in \mathbb{R}^m,\lambda\in \mathbb{R}}  {\lambda}\\
    \text{s.t.}&
    \begin{cases}
        { h(\bm{a},\bm{x}_0)\geq \epsilon_0,}\\
        h'(\bm{a},\bm{x}_i)\geq 0, i=1,\ldots,N,\\
        \lambda \in [0,\overline{\lambda}]; a_l \in [-U_{a_l}, U_{a_l}], l=1,\ldots, m, 
    \end{cases}
    \end{split}
\end{equation}
where $h'(\bm{a},\bm{x}_i)=\sum_{j=1}^M w^L_{j}(\bm{x}_i) h(\bm{a},\bm{f}(\bm{x}_i,\bm{u}_j))\geq \gamma h(\bm{a},\bm{x}_i)-\lambda$. \eqref{N_N_u_r_l1} is a scenario optimization to the following uncertain convex optimization: 
    \begin{equation*}
\label{N_N_u_r}
  \begin{split}
    &\textstyle{\min_{\bm{a}\in \mathbb{R}^m,\lambda \in \mathbb{R}}  \lambda}\\
    \text{s.t.}&
    \begin{cases}
   { h(\bm{a},\bm{x}_0)\geq \epsilon_0,}\\
        \sum_{j=1}^M w_{j}^L(\bm{x}) h(\bm{a},\bm{f}(\bm{x},\bm{u}_j))\geq \gamma h(\bm{a},\bm{x})-\lambda, \forall \bm{x}\in \mathcal{X}, \\
        \lambda \in [0,\overline{\lambda}];~ a_l \in [-U_{a_l}, U_{a_l}];~ l=1,\ldots, m. \\
    \end{cases}
    \end{split}
\end{equation*}

Let $(\bm{a}^{*}(\mathbb{X}),\lambda^*(\mathbb{X}))$ be an optimal solution to \eqref{N_N_u_r_l1}. If $\lambda^*(\mathbb{X})=0$, we proceed with addressing the following optimization \eqref{barrier_convex11} that maximizes the size of the set $\widetilde{\mathbb{S}}=\{\bm{x}\in \mathcal{X}\mid h(\bm{a},\bm{x})\geq 0\}$ according to a tie-break rule as suggested in \citep{calafiore2006scenario}:
\begin{align}
    \label{barrier_convex11}
    &\textstyle\min_{\bm{a}\in \mathbb{R}^m}  {\sum_{i=1}^{N'}h(\bm{a},\bm{x}'_i)}\\
    \text{s.t.}&
    \begin{cases}
        { h(\bm{a},\bm{x}_0)\geq \epsilon_0,}\\
        \sum_{j=1}^M w_{j}^L(\bm{x}_i) h(\bm{a},\bm{f}(\bm{x}_i,\bm{u}_j))\geq \gamma h(\bm{a},\bm{x}_i),  \\
        a_l \in [-U_{a_l}, U_{a_l}]; ~
         l=1,\ldots, m; ~i=1,\ldots, N.
    \end{cases}\notag
\end{align}

    \begin{theorem}
\label{theo:1}
   Let $\alpha, \beta\in (0,1)$, $(\bm{a}^{*}(\mathbb{X}),\lambda^*(\mathbb{X}))$ be the optimal solution to \eqref{N_N_u_r_l1} (if $\lambda^*(\mathbb{X})=0$, $\bm{a}^{*}(\mathbb{X})$ is obtained via solving \eqref{barrier_convex11} according to a tie-break rule) with $\mathbb{X}=\{\bm{x}_i\}_{i=1}^m \stackrel{\text{i.i.d.}}{\sim}\textnormal{P}$, and define the computational procedure as in \eqref{cp}.
Then, if $N$ and $m$ satisfy $\alpha\geq {\frac{2}{N}(\ln{\frac{1}{\beta}}+m+1)}$,
the following guarantee holds:
\begin{equation*}
    \begin{split}
 &   \textnormal{P}_{\mathbb{X}}\left[ 
 \begin{split}
 \Xi(\mathbb{X})\Rightarrow  
 \begin{split}
 &\textstyle\textnormal{P}[\exists \bm{u}\in \mathbb{U}, \bm{f}(\bm{x},\bm{u})\in \widetilde{\mathbb{S}}\mid \bm{x}\in \widetilde{\mathbb{S}}]\\
 &\geq 1-\frac{1}{\textnormal{P}[\bm{x}\in\widetilde{\mathbb{S}}]}\alpha
 \end{split}
 \end{split}
 \right]\geq 1-\beta,
    \end{split}
\end{equation*}
where $\widetilde{\mathbb{S}}=\{\bm{x}\in \mathcal{X}\mid h(\bm{a}^{*}(\mathbb{X}),\bm{x})\geq 0\}$ and $\textnormal{P}_{\mathbb{X}}:=\textnormal{P}^N$.
    \end{theorem}
\begin{pf}
According to Proposition \ref{formal_PAC}, we have that, 
\[\textnormal{P}_{\mathbb{X}}\left[\textnormal{P}\left[\begin{split}
&\sum_{j=1}^M w_{j}^L(\bm{x}) h(\bm{a}^*(\mathbb{X}),\bm{f}(\bm{x},\bm{u}_j))\\
&\geq \gamma h(\bm{a}^*(\mathbb{X}),\bm{x})-\lambda^*(\mathbb{X})
\end{split}
\right]\geq 1-\alpha\right]\geq 1-\beta.\]
On the other hand, since 
\begin{equation*}
\label{max_in}
\begin{split}
   \max_{\bm{u}\in \mathbb{U}} {h(\bm{a}^*(\mathbb{X}),\bm{f}(\bm{x},\bm{u}))} \geq &\max\{h(\bm{a}^*(\mathbb{X}),\bm{f}(\bm{x},\bm{u}_j))\}_{j=1}^M \\
   \geq &\sum_{j=1}^M w_{j}^L(\bm{x}) h(\bm{a}^*(\mathbb{X}),\bm{f}(\bm{x},\bm{u}_j)) 
\end{split}
\end{equation*}
for $\bm{x}\in \mathcal{X}$, we have
\[\textnormal{P}_{\mathbb{X}}\left[\textnormal{P}\left[\begin{split}
&\max_{\bm{u}\in \mathbb{U}} {h(\bm{a}^*(\mathbb{X}),\bm{f}(\bm{x},\bm{u}))}\\
&\geq \gamma h(\bm{a}^*(\mathbb{X}),\bm{x})-\lambda^*(\mathbb{X})
\end{split}
\right]\geq 1-\alpha\right]\geq 1-\beta.\]
Finally, the conclusion can be obtained via following the proof of Proposition \ref{pro3}. \qed
\end{pf}

The method mainly involves solving linear programs, with complexity depending on 
\(\alpha\), \(\beta\), and \(m\), not the state space size. This makes it practical 
for large systems, and engineering knowledge can help choose suitable barrier templates. The condition 
\(\alpha \ge \frac{2}{N}\big(\ln\frac{1}{\beta} + m + 1\big)\) 
shows that \(\beta = 0\) would require infinite samples. In practice, a small 
\(\beta\) (e.g., \(10^{-20}\)) only slightly increases the required samples. 
With this, we can be almost certain that every state 
\(\bm{x} \in \widetilde{\mathbb{S}}\) has a path to a safe state in \(\widetilde{\mathbb{S}}\) with probability 
at least \(1 - \frac{\alpha}{\textnormal{P}[\bm{x} \in \widetilde{\mathbb{S}}]}\).  In Alg. \ref{alg1}, we use a training dataset $X^L$ with the same size as the sampling set $\mathbb{X}$ used in the scenario optimization so that the learned weighting function $w_j^L$ aligns with the $N$ sampled states. Studying other sizes of $X^L$ is left for future work.

\section{Examples}
\label{sec:ex}
In this section, we demonstrate the effectiveness of our approach through a series of examples. All computations are conducted on a Windows machine equipped with an Intel i7-13700H CPU and 32 GB of RAM.

In experiments, we parameterize $h_1(\bm{a},\bm{x})$ as a polynomial of degree $d$. The parameters are set as follows: $\beta= 10^{-20}$, $N^\prime=10^3$, $K=5$, $K^{'}=10$, $U_{a_l} = 10^3$, $C=-1$, $\epsilon_0=10^{-6}$, and $\bm{x}_0$ is set to the origin. To estimate $\textnormal{P}_{\bm{x}}[\bm{x}\in\widetilde{\mathbb{S}}]$, we employ the Monte Carlo method, sampling $10^6$ states from $\mathcal{X}$ uniformly and computing $\frac{1}{10^6}\sum_{i=1}^{10^6} 1_{\widetilde{\mathbb{S}}}(\bm{x}_i)$. 
\subsection{Case Studies}
\begin{example}[PAC SISs]
\label{ex:1} Consider the VanderPol oscillator in \citep{henrion2013convex},
    \begin{equation*}
        \begin{cases}
        x(t+1)=x(t) + 0.01(-2y(t)),\\
        y(t+1)=y(t) + 0.01(0.8x(t)-10(y(t)-0.21)y(t)),
        \end{cases}
    \end{equation*}
    where the safe set is $\mathcal{X}=\{\,(x,y)^{\top}\mid x^2+y^2-1.1 < 0\,\}$. 

    \begin{figure}[t]
    \centering
    \subfigure[$\alpha=0.3$]{
    \includegraphics[width=0.3\linewidth]{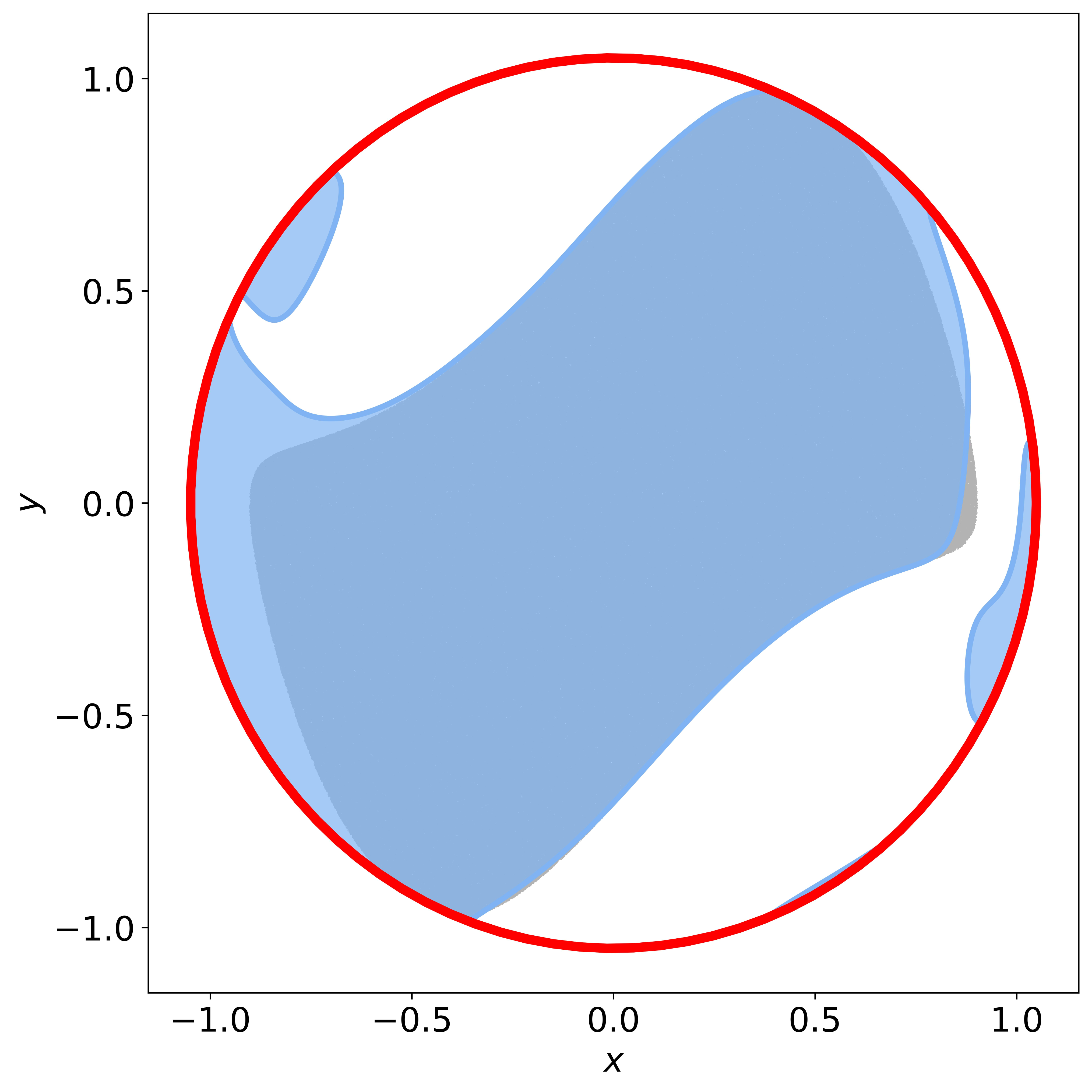}
    \label{fig:vanderpol1}
    }
    \subfigure[{$\alpha=0.1$}]{
    \includegraphics[width=0.3\linewidth]{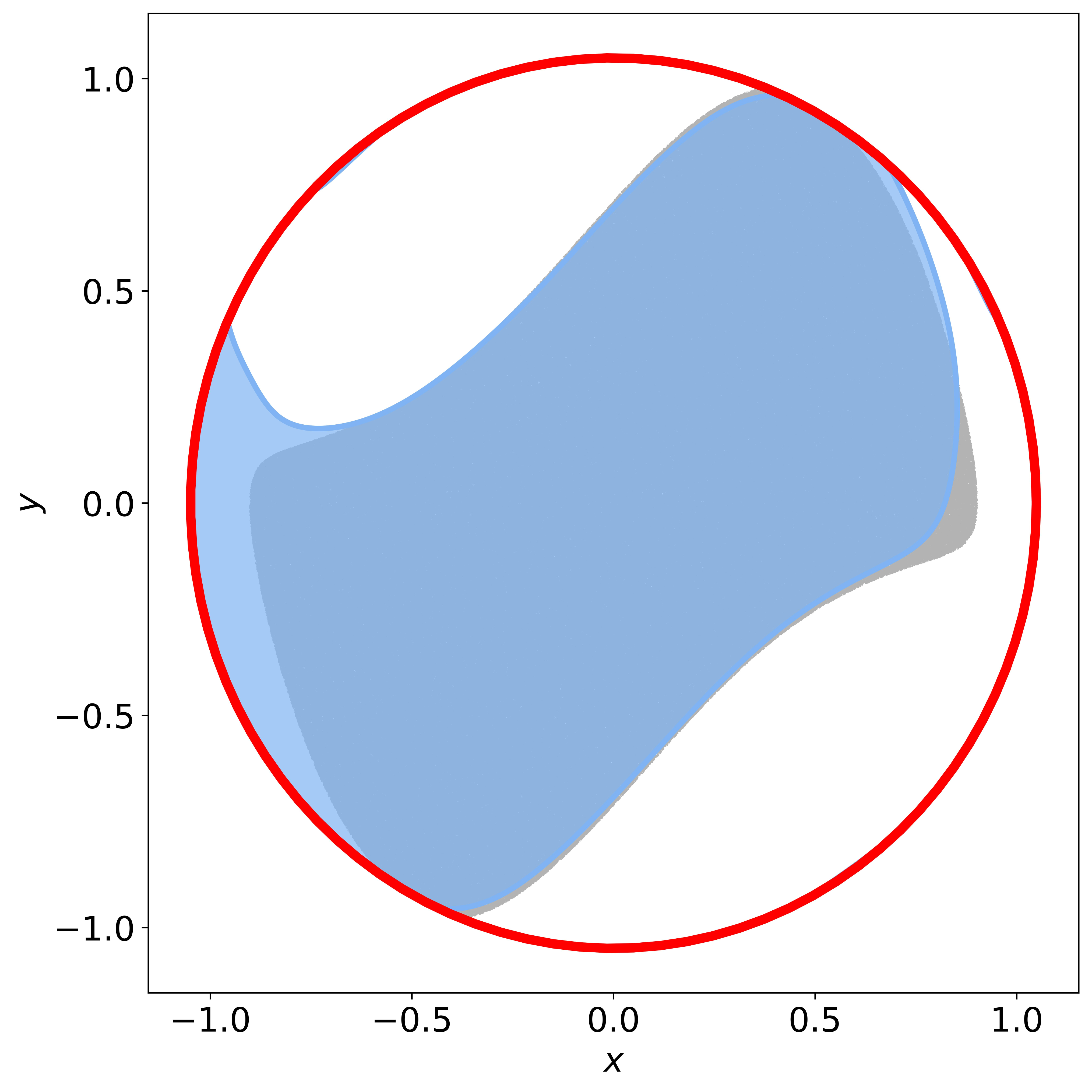}
    \label{fig:vanderpol2}
    }
    \subfigure[{$\alpha=0.05$}]{
    \includegraphics[width=0.3\linewidth]{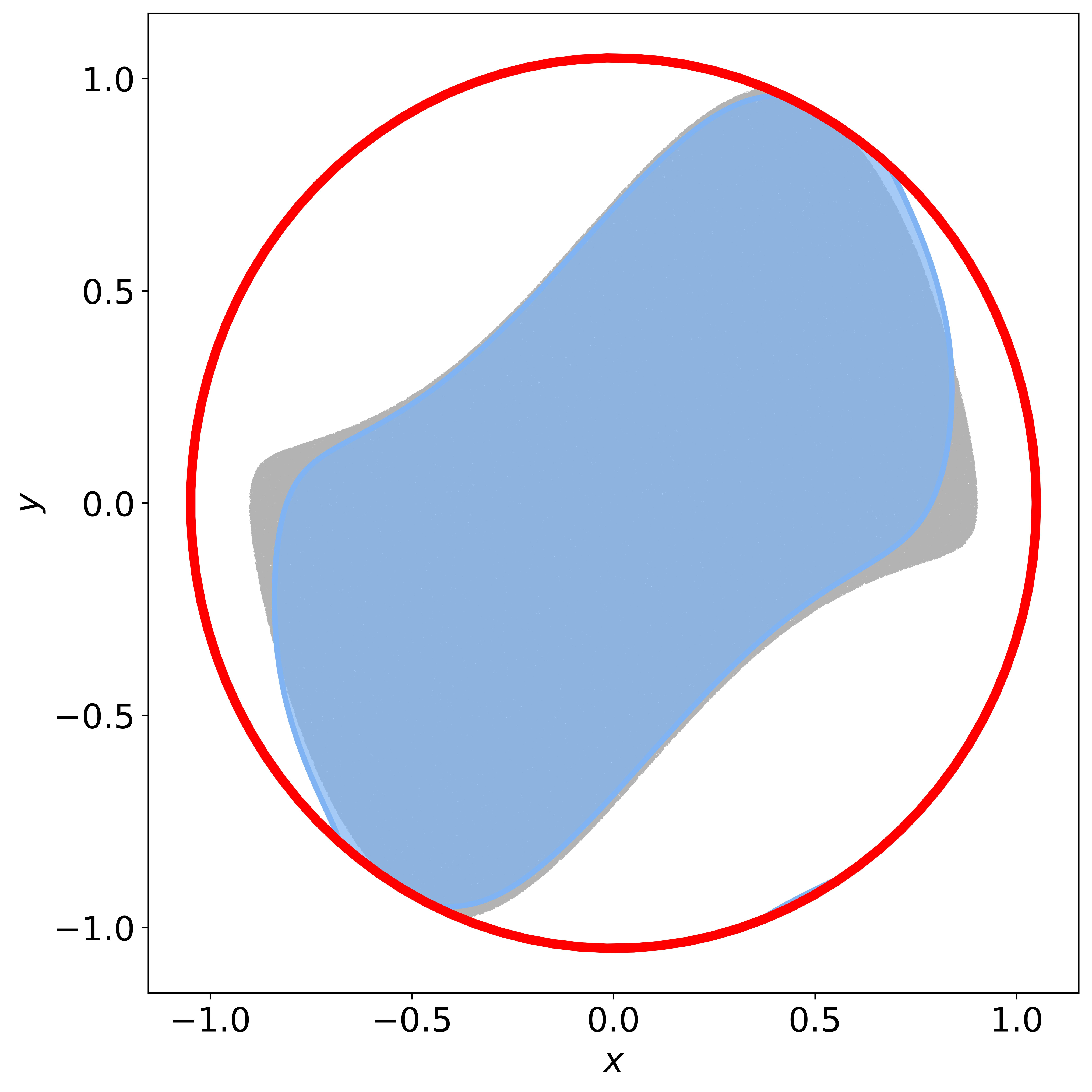}
    \label{fig:vanderpol3}
    }
    \subfigure[{$\alpha=0.01$}]{
    \includegraphics[width=0.3\linewidth]{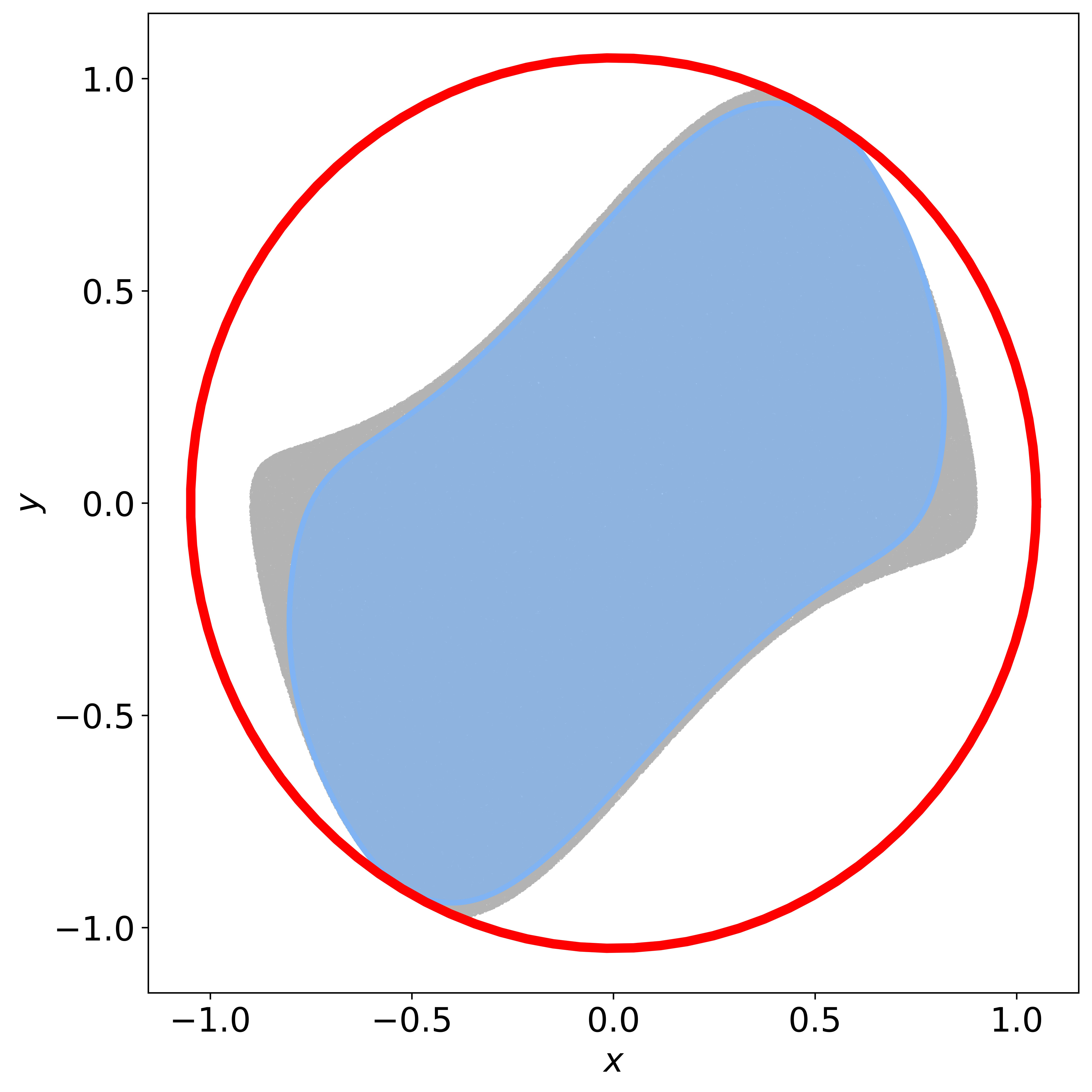}
    \label{fig:vanderpol4}
    }
    \caption{\textcolor{red}{Red} curves: the boundary of $\mathcal{X}$. \textcolor[RGB]{127,178,244}{Blue} region: the calculated PAC SIS, and \textcolor[RGB]{127,127,127}{gray} region: the SIS obtained using the Monte Carlo method.}
    \label{fig:vanderpol_alpha}
\end{figure}
\begin{table}[t]
\centering
\caption{Parameters and Results}
\setlength{\tabcolsep}{0.8mm}{
    \begin{tabular}{*{10}{c}}
\toprule
EX  & $\alpha$ & $N$ & $M$ & $d$ & $\epsilon$ & Time & $\textnormal{P}_{\bm{x}}[\bm{x}\in\widetilde{\mathbb{S}}]$ & $1-\frac{\alpha}{\textnormal{P}_{\bm{x}}[\bm{x}\in\widetilde{\mathbb{S}}]}$ \\ \midrule
\multirow{5}{*}{\ref{ex:1}} & $0.3$ & {$921$} & $-$ & $12$ & $-$ & {$0.52$} & {$0.6646$} & {$0.5486$}  \\\cmidrule(lr){2-9} 
 & $0.1$ & {$2762$} & $-$ & $12$ & $-$ & {$0.76$} & {$0.6231$} & {$0.8395$}  \\\cmidrule(lr){2-9} 
 & $0.05$ & {$5523$} & $-$ & $12$ & $-$ & {$1.16$} & {$0.5629$} & {$0.9112$}  \\\cmidrule(lr){2-9} 
 & $0.01$ & {$27611$} & $-$ & $12$ & $-$ & {$2.91$} & {$0.5352$} & {$0.9813$}  \\\midrule
 \multirow{8}{*}{\ref{ex:3}} 
  & $0.01$ & {$13611$} & $9$ & $5$ & $0.7$ & {$3.49$} & {$0.5984$} & {$0.9833$} \\\cmidrule(lr){2-9}
  & $0.01$ & {$13611$} & $9$ & $5$ & $0.5$ & {$3.42$} & {$0.6927$} & {$0.9856$} \\\cmidrule(lr){2-9}
 & $0.01$ & {$13611$} & $9$ & $5$ & $0.3$ & {$3.36$} & {$0.6529$} & {$0.9847$} \\\cmidrule(lr){2-9}
 & $0.01$ & {$13611$} & $9$ & $5$ & $0.1$ & {$3.54$} & {$0.5746$} & {$0.9826$} \\\cmidrule(lr){2-9}
& $0.01$ & {$13611$} & $9$ & $5$  & $0.0$ & {$3.41$} & {$0.5526$} & {$0.9819$} \\\midrule
\ref{ex:4} & $0.01$ & {$27611$} & $8$ & $2$ & $0.1$ & {$11.9$} & {$0.7931$} & {$0.9874$} \\
\bottomrule
\end{tabular}}
\label{table}
\end{table}

In this example, we fix $\beta$ and other parameters while varying $\alpha$ to study its effect on PAC SISs (Fig. \ref{fig:vanderpol_alpha}, Table \ref{table}). The maximal SIS is estimated via Monte Carlo: $10^6$ states are sampled uniformly from the safe set, and those that remain for 500 steps form the maximal SIS (gray region).

The PAC SIS has the largest volume at $\alpha=0.3$, but also includes the most states that may leave the safe set, which aligns with the smallest $1-\frac{\alpha}{\textnormal{P}_{\bm{x}}[\bm{x}\in \widetilde{\mathbb{S}}]}$ value.  As $\alpha$ decreases, the PAC SIS shrinks, reducing risky states because more samples are used. When all states are considered, the PAC SIS becomes a true SIS. For $\alpha=0.01$, the PAC SIS is a subset of the maximal SIS.
\end{example}

\begin{example}[PAC CSISs]
\label{ex:3} Consider the discrete model adapted from \citep{tan2008stability},
    \begin{equation*}
        \begin{cases}
        x(t+1)=x(t) + 0.01(-0.41x(t)-1.05y(t)-2.3x^2(t)\\
        \quad\quad\quad\quad\quad\quad\quad\quad\quad-0.56x(t)y(t)-x^3(t)+x(t)u_1(t)),\\
        y(t+1)=y(t) + 0.01(1.98x(t)+x(t)y(t)+y(t)u_2(t)),
        \end{cases}
    \end{equation*}
where $\mathcal{X} = \{(x, y)^\top \mid -3\leq x,y \leq 3\}$, and $\mathbb{U} = \{(u_1, u_2)^\top \mid -1 \leq u_1, u_2 \leq 1\}$.

In this example, we evaluate Alg. \ref{alg1} under five $\epsilon$ settings (Table \ref{table}). The first iteration leads to a conservative PAC CSIS with $\textnormal{P}_{\bm{x}}[\bm{x}\in\widetilde{\mathbb{S}}]=0.0786$. Subsequent iterations expand the PAC CSIS under each $\epsilon$. When $\epsilon=0.5$, the PAC CSIS is largest, with the highest $\textnormal{P}_{\bm{x}}[\bm{x}\in \widetilde{\mathbb{S}}]$, and $1-\frac{\alpha}{\textnormal{P}_{\bm{x}}[\bm{x}\in \widetilde{\mathbb{S}}]}$ also reaches its maximum. These results outperform the purely greedy case ($\epsilon=0$), demonstrating the benefit of the $\epsilon$-greedy strategy in computing a larger PAC CSIS.
\end{example}

\begin{example}[PAC CSISs]
\label{ex:4} Consider the Lorenz model of dimension $12$ adapted from \citep{lorenz1996predictability},
\begin{equation*}
    \begin{cases}
    x_i(t+1) = x_i(t) + 10^{-2}\big((x_{i+1}(t) - x_{i-2}(t))x_{i-1}(t)\\
    \hspace{3cm}\quad\quad -x_i(t)+2\big),\ i = 1,2, ..., 9,\\
    x_i(t+1) = x_i(t) + 10^{-2}\big((x_{i+1}(t) - x_{i-2}(t))x_{i-1}(t)\\
    \quad\quad\quad\quad\quad\quad -x_i(t)+2+u_{i-9}(t)\big),\ i = 10,11,12,
    \end{cases}
\end{equation*}
where $x_{-1} = x_{11}$, $x_{0} = {12}$, and $x_{13} = x_{1}$. The safe set is $\mathcal{X} = \{\bm{x} \mid -15\leq x_0, \ldots, x_{12} \leq 15\}$ and the control set is $\mathbb{U} = \{\bm{u} \mid -10 \leq u_1, u_2, u_3 \leq 10\}$. For high-dimensional systems such as this example, our approach effectively computes a PAC CSIS, as shown in Table~\ref{table}, showcasing the scalability of our method to complex systems. Consistent with earlier discussion, the complexity of our method is independent of the dimension of the state space, enabling it to scale effectively to high-dimensional scenarios.
\end{example}

\subsection{Comparisons}

To our knowledge, no existing methods directly compute PAC SISs and CSISs for black-box discrete-time nonlinear systems. Thus, we compare our approach with two model-based baselines—SOS programming and a Bellman-equation method—using Examples \ref{ex:1} and \ref{ex:c1}--\ref{ex:c4}. The SOS method handles uncontrolled polynomial systems by turning \eqref{barrier0} into a semidefinite program. The Bellman method, adapted from \citep{xue2021robust}, computes the maximal SIS through value iteration.

\begin{table}[h!]
\centering
\caption{Comparison with Existing Methods}
\setlength{\tabcolsep}{0.8mm}{
\begin{tabular}{*{9}{c}}
\toprule
\multirow{2}{*}{EX} & Monte Carlo & \multicolumn{2}{c}{SOS} & \multicolumn{2}{c}{Bellman} & \multicolumn{3}{c}{PCSIS (ours)}\\
\cmidrule(lr){2-2}\cmidrule(lr){3-4}\cmidrule(lr){5-6}\cmidrule(lr){7-9}
& $\textnormal{P}_{\bm{x}}$ & Time & $\textnormal{P}_{\bm{x}}$ & Time & $\textnormal{P}_{\bm{x}}$ & Time & $\textnormal{P}_{\bm{x}}$ & $1-\frac{\alpha}{\textnormal{P}_{\bm{x}}}$ \\\midrule
\ref{ex:1} & $0.5906$ & $4.25$ & $0.5761$ & $38.2$ & $0.4419$ & {$2.91$} & {$0.5352$} & {$0.9813$} \\
\ref{ex:c1} & $0.9824$ & $2.16$ & $0.9325$ & $1.21$ & $0.9683$ & {$0.71$} & {$0.9793$} & {$0.9949$} \\
\ref{ex:c2} & $0.4374$ & - & - & $2.86$ & $0.4090$ & {$2.14$} & {$0.4262$} & {$0.9765$} \\
\ref{ex:c4} & $0.4079$ & - & - & - & - & {$18.7$} & {$0.3783$} & {$0.9736$} \\
\bottomrule
\end{tabular}}
\label{table_comp}
\end{table}

Table \ref{table_comp} reports the results. Computation time is in seconds, and $\textnormal{P}_{\bm{x}}=\textnormal{P}_{\bm{x}}[\bm{x}\in\widetilde{\mathbb{S}}]$ measures the SIS volume. A dash (“–”) indicates that no solution was obtained. We also list the Monte-Carlo estimate of the maximal SIS volume for reference. Our method and the SOS approach produce SISs of similar size, but our method is faster when the SOS solver succeeds. SOS can also fail due to numerical issues (e.g., Example \ref{ex:c2}). The Bellman-equation method yields comparable SISs but lacks formal guarantees, is more computationally expensive, and scales poorly. In contrast, our method provides PAC guarantees and remains efficient even for higher-dimensional systems such as Example \ref{ex:c4}, where the other methods become infeasible.

\section{Conclusion}
\label{sec:con}
This paper addresses the problem of identifying CSISs for black-box discrete-time systems. We introduce the concept of a PAC CSIS, which quantifies, with a given confidence level, the fraction of states for which a control input exists that keeps the system within the set at the next time step. To compute PAC CSISs from data, we propose a linear programming approach and demonstrate its effectiveness through several numerical examples.

For future work, we plan to extend the nested PAC characterization framework in \cite{xue2020pac} to compute CSISs and controlled reach–avoid sets for black-box discrete-time stochastic systems \cite{abate2008probabilistic, xue2024sufficient}. We also aim to integrate our method with safe reinforcement learning to enable more reliable learning-based control strategies.

\begin{ack}
We extend our sincere gratitude to Dr. Dominik Wagner and Professor Luke Ong of Nanyang Technological University, Singapore, for their insightful discussions that significantly enhanced this work.

This work was partially supported by the National Research Foundation, Singapore, under its RSS Scheme (NRF-RSS2022-009), the CAS Pioneer Hundred Talents Program, the National Research Foundation, Singapore, and DSO National Laboratories under the AI Singapore Programme (AISG Award No: AISG2-RP-2020-017). 
\end{ack}

\bibliography{ifacconf}            

\renewcommand\thesubsection{\Alph{subsection}}

\section*{Appendix}

In the comparative experiment, each run was limited to 1 hour. The PCSIS parameters are given in Tab. \ref{table_para}. For each example, the SOS method used the same polynomial degree and $\gamma$ as PCSIS, and the Bellman-equation method sampled a similar number of states. In the Bellman-equation method, $\alpha = 0.1$; $\epsilon = 10^{-50}$ for Example \ref{ex:1}, and $\epsilon = 10^{-20}$ for Examples \ref{ex:c1}–\ref{ex:c4}. The meanings of $\alpha$ and $\epsilon$ are detailed in \citep{xue2021robust}.

\begin{table}[htbp]
\centering
\caption{Parameters in Experiments}
\setlength{\tabcolsep}{3.0mm}{
    \begin{tabular}{*{5}{c}}
\toprule
Example& $\alpha$ & $N$ & $d$ & $\gamma$ \\ \midrule
\ref{ex:c1} & $0.005$ & {$24821$} & $4$ & $0.9$  \\\midrule
\ref{ex:c2} & $0.01$ & {$27611$} & $12$ & $0.9999$  \\\midrule
\ref{ex:c4} & $0.01$ & {$51411$} & $4$ & $0.9999$  \\
\bottomrule
\end{tabular}}
\label{table_para}
\end{table}

Examples \ref{ex:c1}--\ref{ex:c4} are presented as follows.
\begin{example}
\label{ex:c1} Consider the predator-prey model,
    \begin{equation*}
        \begin{cases}
        x(t+1)=0.5x(t) - x(t)y(t),\\
        y(t+1)=-0.5y(t) + x(t)y(t),
        \end{cases}
    \end{equation*}
    where the safe set is $\mathcal{X}=\{\,(x,y)^{\top}\mid x^2+y^2-1 < 0\,\}$. 
\end{example}

\begin{example}
\label{ex:c2} Consider the following nonlinear system
    \begin{equation*}
        \begin{cases}
        x(t+1)=2x^2(t) + y(t),\\
        y(t+1)=-2(2x^2(t) + y(t))^2-0.8x(t),
        \end{cases}
    \end{equation*}
    where the safe set is $\mathcal{X}=\{\,(x,y)^{\top}\mid -1 \leq x,y \leq 1\,\}$. 
\end{example}

\begin{example}
\label{ex:c4} Consider the following nonlinear system from \citep{edwards2024fossil}:
    \begin{equation*}
        \begin{cases}
        x_1(t+1)=x_1(t) + 0.01(x_2(t)x_4(t)-x_1^3(t)),\\
        x_2(t+1)=x_2(t) + 0.01(-3x_1(t)x_4(t)-x_2^3(t)),\\
        x_3(t+1)=x_3(t) + 0.01(-x_3(t)-3x_1(t)x_4^3(t)),\\
        x_4(t+1)=x_4(t) + 0.01(-x_4(t)+x_1(t)x_3(t)),\\
        x_5(t+1)=x_5(t) + 0.01(-x_5(t)+x_6^3(t)),\\
        x_6(t+1)=x_6(t) + 0.01(-x_5(t)-x_6(t)+x_3^4(t)),\\
        \end{cases}
    \end{equation*}
    where $\mathcal{X}=\{\,(x_1,\ldots,x_6)^{\top}\mid -0.5 \leq x_1,\ldots,x_6 \leq 2\,\}$. 
\end{example}

\end{document}